\documentclass[letter,reqno]{a_preamble}
\title{
A note on the quantum Wielandt inequality
}
\date{}

\usepackage{authblk}
\author{Owen Ekblad\orcidlink{0009-0006-0834-0327}\thanks{ekbladow@msu.edu}}
\affil{Michigan State University, Department of Mathematics}

\begin{document}

\pagenumbering{arabic}
\lhead{\thepage}
\maketitle

\begin{abstract}
In this note, we prove that the index of primitivity of any primitive unital Schwarz map is at most $2(D-1)^2$, where $D$ is the dimension of the underlying matrix algebra.
This inequality was first proved by Rahaman for Schwarz maps which were both unital and trace preserving. 
As we show, the assumption of unitality is basically innocuous, but
in general not all primitive unital Schwarz maps are trace preserving.
Therefore, the precise purpose of this note is to showcase how to apply the method of Rahaman to unital primitive Schwarz maps that don't preserve trace. 
As a corollary of this theorem, we show that the index of primitivity of any primitive 2-positive map is at most $2(D-1)^2$, so in particular this bound holds for arbitrary primitive completely positive maps.
We briefly discuss of how this relates to a conjecture of Perez-Garcia, Verstraete, Wolf and Cirac. 
\end{abstract}





%
%
\hspace{-5.2mm}\textbf{Acknowledgments}
The author completed this work while partially supported by research assistantships from Professor Jeffrey Schenker with funding from the US National Science Foundation under Grant No. DMS-2153946.


 \section{Introduction}
%
%
%
%
%
%
%
%
%
%
Quantum Wielandt inequalities have been studied, both implicitly and explicitly, in various contexts for many years, but have received more attention in recent years \cite{Michaek2019QuantumRevisited, Rahaman2020AInequality, Jia2024AInequality}, in large part due to pioneering work by Perez-Garcia, Verstraete, Wolf and Cirac on matrix product states (MPS) \cite{Perez-Garcia2007MatrixRepresentations}, and have since been studied both for their relationship with MPS and with quantum information theory \cite{Sanz2010AInequality}.
In this short note, we prove an optimal order quantum Wielandt inequality for general primitive maps satisfying the Schwarz inequality (defined below) without any additional assumptions, thereby extending a theorem of Rahaman \cite{Rahaman2020AInequality}.
Let $\matrices$ denote the space of $D\times D$ matrices with entries in $\mbC$, let $\mbP_D$ denote the set of positive semidefinite matrices, and let $\mbP_D^{\circ}$ denote the set of positive definite matrices.
Say a linear map $\phi:\matrices\to\matrices$ is primitive if there exists $n\in\mbN$ such that $\phi^n\seq{\mbP_D\setminus\set{0}}\subseteq\mbP_D^{\circ}$, and say $\phi$ is a Schwarz map if $\phi(a^*a)\geq \phi(a)^*\phi(a)$ for all $a\in\matrices$, where $(\cdot)^*$ denotes the conjugate transpose.
For a primitive map $\phi$, we let $\primind{\phi}$ denote the smallest $n\in\mathbb{N}$ so that $\phi^n\seq{\mbP_D\setminus\set{0}}\subseteq\mbP_D^{\circ}$, and call $\primind{\phi}$ the index of primitivity of $\phi$. 
Our main theorem is the following. 
\begin{thmx}\label{Thm:Main_theorem_introduction}
    $\primind{\phi}\leq 2(D-1)^2$ for primitive unital Schwarz maps $\phi:\matrices\to\matrices$.
\end{thmx}
This bound was proved by Rahaman in \cite{Rahaman2020AInequality} under the assumption that $\phi$ was a both unital \textit{and} trace preserving Schwarz map.
As we show below (Proposition \ref{Prop:All_Schwarz_maps_are_basically_unital}), the assumption that a primitive Schwarz map is unital is, for our purposes, basically innocuous, but there are many examples of unital Schwarz maps which are not trace preserving, and therefore the main purpose of this note is to demonstrate that the assumption of trace preservation in \cite{Rahaman2020AInequality} is unnecessary. 
Let us state as a theorem the following corollary of Theorem \ref{Thm:Main_theorem_introduction} for a large class of positive maps.
Recall a linear map $\phi:\matrices\to\matrices$ is called $2$-positive if $\phi\otimes\operatorname{Id}_2:\matrices\otimes\mbM_2\to\matrices\otimes\mbM_2$ is positive, where $\operatorname{Id}_2:\mbM_2\to\mbM_2$ is the identity map. 
\begin{thmx}\label{Thm:Main_theorem_compos}
    $\primind{\phi}\leq 2(D-1)^2$ for primitive $2$-positive maps $\phi:\matrices\to\matrices$.
\end{thmx}
Notice there is no requirement that $\phi$ be unital or trace preserving in Theorem \ref{Thm:Main_theorem_compos}.
Examples of $2$-positive maps include all completely positive maps, hence Theorem \ref{Thm:Main_theorem_compos} applies to any completely positive maps, but there are many examples of $2$-positive maps that are not completely positive \cite{Cho1992GeneralizedAlgebra}.
Not all $2$-positive maps are Schwarz maps, however, for the simple reason that all Schwarz maps are contractions \cite[Ch. 2]{Strmer2013PositiveAlgebras}; conversely, not all unital Schwarz maps are $2$-positive, per \cite[Appendix A]{Choi1980SomeC-algebras}.
The order of our bound is indeed optimal, since there exist primitive completely positive maps $\phi_D:\matrices\to\matrices$ for which $\primind{\phi_D} = O(D^2)$ \cite{Sanz2010AInequality}. 
In general, we do not know if this inequality is sharp, although, in light of the similarity between the above bound and the optimal \textit{classical} Wielandt inequality---which is $(D-1)^2 + 1$ \cite{Wielandt1950UnzerlegbareMatrizen}---it seems interesting to ask whether the multiplicative factor of 2 in our bound may be lowered.

\section{Schwarz maps}
Let us start with some notation and a general discussion of Schwarz maps.
To denote that $x\in\mbP_D$, write $x\geq 0$, and to denote $x\in\mbP_D^{\circ}$, write $x > 0$. 
Let $a^*$ denote the conjugate transpose of a matrix $a\in\matrices$. 
Call $p\in\matrices$ a projection if $p\in\mbP_D$ and $p^2 = p$. 
For two projections $p, q\in\matrices$, write $p\sim q$ if $\operatorname{Tr}(p) = \operatorname{Tr}(q)$.
Let $\mbI\in\matrices$ denote the identity.
For $z>0$, let $\inner{\cdot}{\cdot}_z$ denote the inner product on $\matrices$ defined by 
\begin{equation}
    \inner{a}{b}_z := \operatorname{Tr}\seq{za^*b},
\end{equation}
where $\operatorname{Tr}$ denotes the usual trace.
Because $z > 0$, this is a nondegenerate inner product making $\matrices$ into a Hilbert space.
In the case that $z = \mbI$, simply write $\inner{\cdot}{\cdot}$ to denote $\inner{\cdot}{\cdot}_\mbI$.
For $\phi:\matrices\to\matrices$ linear, let $\phi^*:\matrices\to\matrices$ denote the unique linear map such that $\inner{\phi(a)}{b} = \inner{a}{\phi^*(b)}$ for all $a, b\in\matrices$.
Recall some basic notions of maps $\matrices\to\matrices$. 
\begin{definition}
    Let $\superops$ denote the set of linear maps $\phi:\matrices\to\matrices$ and let $\operatorname{Id}_D\in \superops$ denote the identity operator.
    Fix $\phi\in\superops$.
    \begin{itemize}
        \item $\phi$ is called {positive} if $\phi\seq{\mbP_D}\subseteq\mbP_D$.

        \item $\phi$ is called $n$-positive if $\phi\otimes\operatorname{Id}_n:\matrices\otimes\mbM_n\to\matrices\otimes\mbM_n$ is positive. 
        We call $\phi$ completely positive if $\phi$ is $n$-positive for all $n$. 

        \item $\phi$ is called a {Schwarz map} if for all $a\in\matrices$, $\phi(a^*a)\geq \phi(a)^*\phi(a)$. 
        In general, call the equation $\phi(a^*a)\geq \phi(a)^*\phi(a)$ the {Schwarz inequality}.

        \item $\phi$ is called unital if $\phi(\mbI) = \mbI$. 

        \item $\phi$ is called is {strictly positive} if $\phi\seq{\mbP_D\setminus\set{0}}\subseteq\mbP_D^{\circ}$. 

        \item $\phi$ is called is {primitive} if there is $n\in\mbN$ such that $\phi^n$ is strictly positive. 


        \item $\phi$ is called is {fully irreducible} if for all projections $p, q\in\matrices$ with $p\sim q$, if there is $\lambda > 0$ such that $\phi(p)\leq \lambda q$, we have $q, p\in\set{0, 1}$.
    \end{itemize}
\end{definition}
Recall that for any primitive positive map $\phi\in\superops$, there is a unique $z>0$ with $\operatorname{Tr}(z) = 1$ such that $\phi^n(z) = rz$, where $r>0$ is the spectral radius of $\phi$ \cite{Evans1978SpectralC-Algebras}. 
We call this $z$ the Perron-Frobenius eigenvector for $\phi$. 
The following fact is well-known, but we provide a sketch of the proof for completeness.
\begin{lemma}\label{Lem:Prim_strong_mixing}
    Assume $\phi:\matrices\to\matrices$ is a unital primitive Schwarz map. 
    Let $\varrho$ denote the Perron-Frobenius eigenvector of $\phi^*$.
    Then for all $a\in\matrices$, $\lim_{n\to\infty} \phi^n(a) = \inner{\varrho}{a}\mbI$.
\end{lemma}
\begin{proof}
    By viewing $\phi$ as an element of $\mbM_{D^2}$ and decomposing it into its Jordan normal form, then, since by primitivity $\lambda = 1$ is the unique eigenvalue of $\phi$ with $|\lambda| = 1$ \cite{Evans1978SpectralC-Algebras}, it is clear that $\lim_{n\to\infty} \phi^n(a) = c \mbI$ for some $c\in\mbC$. 
    On the other hand, taking Ces\`aro means and noting that $\phi^*(\varrho) = \varrho$, we see that 
    \begin{equation}
        c
            =
        \inner{\varrho}{c\mbI}
            =
        \lim_{N\to\infty}
        \cfrac{1}{N}
        \sum_{n=1}^N
        \inner{\varrho}{\phi^n(a)}
            =
        \lim_{N\to\infty}
        \cfrac{1}{N}
        \sum_{n=1}^N
        \inner{\varrho}{a}
            =
        \inner{\varrho}{a},
    \end{equation}
    which concludes the proof. 
\end{proof}
Notice that we have not asked for our Schwarz maps to be unital. 
In the literature, a map $\phi\in\superops$ is often called a Schwarz map if it satisfies the Schwarz inequality for all $a\in\matrices$ in addition satisfying unitality.
We now show that the assumption of unitality is, in the case of interest to us, essentially artificial.
\begin{prop}\label{Prop:All_Schwarz_maps_are_basically_unital}
    A primitive Schwarz map $\phi\in\superops$ has spectral radius 1 if and only if $\phi$ is unital. 
\end{prop}
\begin{proof}
If $\phi$ is unital, then, because Schwarz maps are contractions \cite[Ch. 2]{Strmer2013PositiveAlgebras}, it is clear that the spectral radius of $\phi$ is equal to 1. 
So assume $\phi\in\superops$ is a primitive Schwarz map with spectral radius $r = 1$. 
    Let $\varrho$ be the Perron-Frobenius eigenvector for $\phi^*$. 
    Then $\phi^*(\varrho) = \varrho$, since the spectral radii of $\phi$ and $\phi^*$ are equal. 
    By primitivity, there is $n\in\mbN$ such that $\phi^n(\mbI) = x > 0$. 
    Because $\phi^n$ is Schwarz map, we have that $x = \phi^n(\mbI)\geq \phi^n(\mbI)^2$.
    Iterating this and using the fact that $y\mapsto y^{1/2}$ is an operator monotone function, we see that 
    \begin{equation}
        x^{1/2^m}\geq \phi^n(\mbI)
    \end{equation}
    for all $m\in\mbN$. 
    Taking the limit as $m\to\infty$ and using the fact that $x>0$, we conclude that $\mbI\geq \phi^n(\mbI)$. 
    Thus, 
    \begin{equation}
        \inner{\varrho}{\mbI - \phi^n(\mbI)} = \inner{\varrho}{\mbI} - \inner{(\phi^*)^n(\varrho)}{\mbI}
        =
        \inner{\varrho}{\mbI}
        -
        \inner{\varrho}{\mbI}
        =
        0,
    \end{equation}
    hence $\phi^n(\mbI) = \mbI$, as $\varrho>0$.
    By the primitivity of $\phi$, this implies that $D^{-1}\mbI$ is the Perron-Frobenius eigenvector of $\phi$, and concludes the proof.
\end{proof}

\section{Quadratic order bound on index of primitivity}
Now, for $\phi\in\superops$ a primitive map, define 
\begin{equation}
    \primind{\phi}
    =
    \min\set{k\in\mbN\,\,:\,\, \phi^k\text{ is strictly positive}}.
\end{equation}
We call $\primind{\phi}$ the {index of primitivity} of $\phi$.  
%
%
%
In this work, we are interested in computing a bound on $\primind{\phi}$ in the case that $\phi\in\superops$ is a primitive unital Schwarz map. 
%
%
For a unital Schwarz map $\phi$, it holds \cite[Ch. 2]{Strmer2013PositiveAlgebras} that
\begin{equation}
   \begin{split}
        \set{a\in\matrices\,\,:\,\, \phi(ab) = \phi(a)\phi(b)\,\text{and}\,
    \phi(ba) = \phi(b)\phi(a)\text{ for all $b\in\matrices$}}&\\
        &\hspace{-65mm}=
    \set{a\in\matrices\,\,:\,\, \phi(a^*a) = \phi(a)^*\phi(a)
    \,\text{and}\,
    \phi(aa^*) = \phi(a)\phi(a)^*}.
   \end{split}
\end{equation}
Let us denote the above set, which is a von Neumann algebra, by $\mcM_\phi$,  and call $\mcM_\phi$ the multiplicative domain of $\phi$. 
In the case that $\phi$ is primitive, we can give an alternative description of $\mcM_\phi$.
\begin{lemma}\label{Lem:Mult_dom_of_prim_unital_Schwarz}
    Let $\phi\in\superops$ be a primitive unital Schwarz map and let $\varrho\in\mbP_D^{\circ}$ be the Perron-Frobenius eigenvector of $\phi^*$. 
    For all $n\in\mbN$,  
    \begin{equation}
        \mcM_{\phi^n}
        =
        \set{
        a\in\matrices\,\,:\,\,
        \inner{a}{a}_\varrho
        =
        \inner{\phi^n(a)}{\phi^n(a)}_\varrho
        \,\text{and}\,
        \inner{a^*}{a^*}_\varrho
        =
        \inner{\phi^n(a)^*}{\phi^n(a)^*}_\varrho
        }.
    \end{equation}
\end{lemma}
\begin{proof}
    If $a\in\mcM_{\phi^n}$, then $\phi^n(a^*a) = \phi^n(a)^*\phi^n(a)$. 
    Because $\phi^*(\varrho) = \varrho$, we have that 
    \begin{align}
        \inner{a}{a}_\varrho
        =
        \operatorname{Tr}\seq{\varrho a^*a}
        =
        \operatorname{Tr}\seq{\seq{\phi^n}^*(\varrho)a^*a}
        &=
        \operatorname{Tr}\seq{\varrho \phi^n(a^*a)}\\
        &=
        \operatorname{Tr}\seq{\varrho\phi^n(a)^*\phi^n(a)}
        =
        \inner{\phi^n(a)}{\phi^n(a)}_\varrho.
    \end{align}
    Arguing the same with $a$ replaced by $a^*$, we see that $\inner{a^*}{a^*}_\varrho
        =
        \inner{\phi^n(a)^*}{\phi^n(a)^*}_\varrho$.
    Conversely, suppose $a\in\matrices$ satisfies $\inner{a}{a}_\varrho
        =
        \inner{\phi^n(a)}{\phi^n(a)}_\varrho$.
    By the Schwarz inequality applied to $\phi^n$, we have that 
    \begin{equation}
        \varrho^{1/2}\phi^n(a^*a)\varrho^{1/2}\geq \varrho^{1/2}\phi^n(a)^*\phi^n(a)\varrho^{1/2}\geq 0.
    \end{equation}
    Thus, because $\varrho>0$, taking trace of this equality and using the fact that  $\inner{a}{a}_\varrho
        =
        \inner{\phi^n(a)}{\phi^n(a)}_\varrho$, 
    we conclude that $\phi^n(a^*a) = \phi^n(a)^*\phi^n(a)$. 
    Replacing $a$ by $a^*$ in the above argument yields that $\phi^n(aa^*) = \phi^n(a)\phi^n(a)^*$ and concludes the proof. 
\end{proof}
As a corollary, we get the following. 
\begin{cor}\label{Cor:Nested_mult_doms}
    Let $\phi\in\superops$ be a primitive unital Schwarz map.
    For all $n\in\mbN$,  $\mcM_{\phi^{n+1}}\subseteq\mcM_{\phi^n}$.
\end{cor}
\begin{proof}
    Let $\varrho$ be the Perron-Frobenius eigenvector of $\phi^*$. 
    In general, by the Schwarz inequality, we have that 
    \begin{equation}
        \varrho^{1/2}\phi^{n+1}\seq{a^*a}\varrho^{1/2}
        \geq 
        \varrho^{1/2}\phi\seq{\phi^n(a)^*\phi^n(a)}\varrho^{1/2}
        \geq 
        \varrho^{1/2}\phi^{n+1}\seq{a}^*\phi^{n+1}\seq{a}\varrho^{1/2}
    \end{equation}
    for all $a\in\matrices$. 
    Thus, if $a\in\mcM_{\phi^{n+1}}$, then Lemma \ref{Lem:Mult_dom_of_prim_unital_Schwarz}, the fact that $\phi^*(\varrho) = \varrho$, and the above string of inequalities gives that 
    \begin{equation}
        \inner{\phi^{n+1}(a)}{\phi^{n+1}(a)}_\varrho 
        =
        \inner{a}{a}_\varrho 
        \geq 
        \inner{\phi^n(a)}{\phi^n(a)}_\varrho 
        \geq 
        \inner{\phi^{n+1}(a)}{\phi^{n+1}(a)}_\varrho.
    \end{equation}
    Thus, $\inner{\phi^n(a)}{\phi^n(a)}_\varrho  = \inner{a}{a}_\varrho$.
    Arguing the same with $a^*$ replacing $a$, we see  $\inner{\phi^n(a)^*}{\phi^n(a)^*}_\varrho  = \inner{a^*}{a^*}_\varrho$, hence $a\in\mcM_{\phi^n}$.
\end{proof}
By the finite-dimensionality of $\matrices$ and the above corollary, there is some $k\in\mbN$ such that
\begin{equation}
    \bigcap_{n\in\mbN}\mcM_{\phi^n}
    =
    \mcM_{\phi^k}.
\end{equation}
Let $\kappa(\phi)$ be the minimal such $k$, and write $\mcM_\phi^\infty$ to denote $\mcM_{\phi^{\kappa(\phi)}}$. 
\begin{lemma}\label{Lem:Stab_mult_dom_of_primitive}
    For a primitive unital Schwarz map $\phi\in\superops$, $\mcM_\phi^\infty = \mbC\mbI$. 
\end{lemma}
\begin{proof}
    Let $\varrho$ be the Perron-Frobenius eigenvector of $\phi^*$.
    Because $\mcM_\phi^\infty$ is a von Neumann algebra, it is the closed span of its projections, so it suffices to show that if $p\in\mcM_\phi^\infty$ is a projection, $p \in\set{0, \mbI}$.
    For $p\in\mcM_\phi^\infty$ a projection, we know that 
    \begin{equation}
        \operatorname{Tr}\seq{\varrho p}
        =
        \inner{p}{p}_\varrho
        =
        \lim_{n\to\infty}
        \inner{\phi^n(p)}{\phi^n(p)}_\varrho 
        =
        \operatorname{Tr}\seq{
        \varrho 
        p
        }^2,
    \end{equation}
    which follows from Lemma \ref{Lem:Prim_strong_mixing}. 
    Thus, $ \operatorname{Tr}\seq{\varrho p}\in\set{0, 1}$. 
    Because $\varrho > 0$ and $\operatorname{Tr}(\varrho) = 1$, this implies $p\in\set{0, \mbI}$, which concludes the proof. 
\end{proof}
\begin{cor}\label{Cor:Prim_implies_fully_irred}
    For any primitive unital Schwarz map $\phi\in\superops$, $\phi^{\kappa(\phi)}$ is fully irreducible. 
\end{cor}
\begin{proof}
    Let $\varrho$ be the Perron-Frobenius eigenvector of $\phi^*$.
   Let $p, q\in\matrices$ be projections with $p\sim q$ and suppose there is $\lambda > 0$ such that $\phi^{\kappa(\phi)}(p)\leq \lambda q$. 
   By \cite[Proposition 3.3]{Rahaman2020AInequality}, this implies $\phi^{\kappa(\phi)}(p)\leq q$. 
   Thus, from $\varrho > 0$ and $\phi^*(\varrho) = \varrho$, we find that $0\leq \varrho^{1/2}\phi^{\kappa(\phi)}(p)\varrho^{1/2}\leq \varrho^{1/2}q \varrho^{1/2}$, so by taking trace and using that $\varrho>0$, we conclude that $\phi^{\kappa(\phi)}(p) = q$. 
   Therefore, $\phi^{\kappa(\phi)}(p)$ is a projection. 
   But then, using that $\phi^*(\varrho) = \varrho$, we discover that 
   \begin{equation}
       \inner{p}{p}_\varrho 
       =
       \inner{\phi^{\kappa(\phi)}(p)}{\phi^{\kappa(\phi)}(p)}_\varrho,
   \end{equation}
   hence $p\in\mcM_\phi^\infty$. 
   By Lemma \ref{Lem:Stab_mult_dom_of_primitive}, this implies $p\in\set{0, \mbI}$, and therefore $q = \phi(p)\in\set{0, \mbI}$ as well. 
\end{proof}
\begin{cor}\label{Cor:Prim_implies_rank_increasing}
    For any primitive unital Schwarz map $\phi\in\superops$ and projection $p$ with $p\not\in\set{0, \mbI}$, we have that $\operatorname{rank}\seq{\phi^{\kappa(\phi)}(p)} > \operatorname{rank}\seq{p}$. 
\end{cor}
\begin{proof}
    The argument follows as in \cite[Theorem III.7]{Rahaman2020AInequality}, but we repeat it here for completeness. 
    For any projection $p\in\mbP_D$, the Schwarz inequality gives that $\phi^{\kappa(\phi)}(p)\geq \phi^{\kappa(\phi)}(p)^2$.
    Because $x\mapsto x^{1/2}$ is an operator monotone function, this gives
    \begin{equation}
       \phi^{\kappa(\phi)}(p)^{1/2^n}\geq \phi^{\kappa(\phi)}(p)
    \end{equation}
    for all $n$. 
    Letting $q$ denote the range projection of $\phi^{\kappa(\phi)}(p)$, taking the limit in $n$ in the above inequality yields that $q\geq \phi^{\kappa(\phi)}(p)$. 
    By Corollary \ref{Cor:Prim_implies_fully_irred}, it must be that $\operatorname{Tr}(q) > \operatorname{Tr}(p)$, as otherwise $p\sim q$ hence $p\in\set{0, \mbI}$ by the full irreducibility of $\phi^{\kappa(\phi)}$.
    But we assumed $p\not\in\set{0, \mbI}$, so indeed $\operatorname{Tr}(q) > \operatorname{Tr}(p)$, i.e., $\operatorname{rank}\seq{\phi^{\kappa(\phi)}(p)} > \operatorname{rank}\seq{p}$.
\end{proof}
\begin{cor}\label{Cor:Bound_on_primind}
    For a primitive unital Schwarz map $\phi\in\superops$, $q(\phi)\leq (D-1)\kappa(\phi)$.
\end{cor}
\begin{proof}
    We show that $\phi^{(D-1)\kappa(\phi)}$ is strictly positive. 
    To do this, it suffices by the spectral theorem to show that $\phi^{(D-1)\kappa(\phi)}(p)>0$ for any rank one projection $p$. 
    This, in turn, follows from Corollary \ref{Cor:Prim_implies_rank_increasing}, which implies
    \begin{equation}
        \operatorname{rank}\seq{\phi^{(D-1)\kappa(\phi)}(p)}
        \geq \operatorname{rank}(p) + D - 1
        =
        D,
    \end{equation}
    completing the proof. 
\end{proof}
Before we can prove our main result, we need one final technical lemma.
\begin{lemma}\label{Lem: Equality_iff_mult_ind_achieved}
    For any primitive unital Schwarz map and $n\in\mbN$, $\mcM_{\phi^n} = \mcM_{\phi^{n+1}}$ if and only if $n \geq \kappa(\phi)$. 
\end{lemma}
\begin{proof}
    Let $\varrho$ be the Perron-Frobenius eigenvector of $\phi^*$.
    The converse direction is clear, so we only need to prove the forward direction. 
    To show this, in turn, it suffices by Lemma \ref{Lem:Stab_mult_dom_of_primitive} to show that if for some $n\in\mbN$ we have $\mcM_{\phi^n} = \mcM_{\phi^{n+1}}$, then $\mcM_{\phi^n} = \mbC\mbI$. 
    To do this, notice that if $\mcM_{\phi^n} = \mcM_{\phi^{n+1}}$, we have $\phi\seq{\mcM_{\phi^n}}\subseteq\mcM_{\phi^n}$. 
    Indeed, by Lemma \ref{Lem:Mult_dom_of_prim_unital_Schwarz} and Corollary \ref{Cor:Nested_mult_doms},
    \begin{equation}
        \inner{\phi^n\seq{\phi(a)}}{\phi^n\seq{\phi(a)}}_\varrho 
        =
        \inner{\phi^{n+1}\seq{a}}{\phi^{n+1}\seq{a}}_\varrho 
        =
        \inner{a}{a}_\varrho
        =
        \inner{\phi(a)}{\phi(a)}_\varrho,
    \end{equation}
    so again by  Lemma \ref{Lem:Mult_dom_of_prim_unital_Schwarz} we have $\phi(a)\in\mcM_{\phi^{n}}$. 
    However, we also know from Lemma \ref{Lem:Mult_dom_of_prim_unital_Schwarz} and the fact that $\varrho>0$ that $\phi\vert_{\mcM_{\phi^n}}$ is injective, and therefore $\phi\seq{\mcM_{\phi^n}} = \mcM_{\phi^n}$, and $\phi\vert_{\mcM_{\phi^n}}$ is invertible. 
    Define $\phi_\varrho^*\in\superops$ by $\phi_\varrho^*(a) := \phi^*(a\varrho)\varrho^{-1}$. 
    Then it is straightforward to check that $\inner{\phi(a)}{b}_\varrho = \inner{a}{\phi^*_\varrho(b)}_\varrho$ for all $a, b\in\matrices$, i.e., $\phi_\varrho^*$ is the adjoint of $\phi$ with respect to the inner product $\inner{\cdot}{\cdot}_\varrho$. 
    We claim that $\phi^*_\varrho\vert_{\mcM_{\phi^n}} = \phi\vert_{\mcM_{\phi^n}}^{-1}$. 
    Indeed,  if we let $\psi:\mcM_{\phi^n}\to\mcM_{\phi^n}$ denote $\phi\vert_{\mcM_{\phi^n}}^{-1}$, for any $a\in\mcM_{\phi^n}$ and $c\in\matrices$ we have
    \begin{align}
        \inner{c}{\phi^*_\varrho(a)}_\varrho 
        =
        \inner{\phi(c)}{a}_\varrho 
        = 
        \inner{\phi(c)}{\phi\seq{\psi(a)}}_\varrho
        &=
        \operatorname{Tr}\seq{
        \varrho 
        \phi(c)
        \phi(\psi(a))
        }\\
        &=
        \operatorname{Tr}\seq{
        \varrho 
        \phi(c\psi(a))
        }
        && \psi(a)\in\mcM_{\phi^n}\subseteq\mcM_\phi\\
        &=
        \operatorname{Tr}\seq{
        \varrho c\psi(a)
        } &&\phi^*(\varrho) = \varrho\\
        &= 
        \inner{c}{\psi(a)}_\varrho.
    \end{align}
    Thus, since $\inner{\cdot}{\cdot}_\varrho$ is a nondegenerate inner product, we conclude that $\psi(a) = \phi_\varrho^*(a)$ for all $a\in\mcM_{\phi^n}$. 
    Thus, $\phi:\mcM_{\phi^n}\to\mcM_{\phi^n}$ is a unitary map with respect to the nondegenerate inner product $\inner{\cdot}{\cdot}_\varrho$.
    In particular, for any $a\in\mcM_{\phi^n}$, we have that 
    \begin{equation}
        \inner{a}{a}_\varrho
        =
        \lim_{m\to\infty}\inner{\phi^m(a)}{\phi^m(a)}_\varrho.
    \end{equation}
    Arguing as in Lemma \ref{Lem:Stab_mult_dom_of_primitive}, we conclude that $\mcM_{\phi^n} = \mbC\mbI$, which concludes the proof. 
\end{proof}
We may now prove Theorem \ref{Thm:Main_theorem_introduction}.
\begin{proof}[Proof of Theorem \ref{Thm:Main_theorem_introduction}]
    By Corollary \ref{Cor:Bound_on_primind}, therefore, it suffices to show that $\kappa(\phi) \leq 2(D - 1)$. 
    By Lemma \ref{Lem: Equality_iff_mult_ind_achieved}, we know that $\kappa(\phi)$ is bounded above by the maximal length of a chain of unital $C^*$-subalgebras 
    \begin{equation}
        \matrices \supseteq \mcA_1 \supsetneq \mcA_2\supsetneq\cdots\supsetneq \mcA_{\ell} = \mbC\mbI. 
    \end{equation}
    By \cite[Lemma 3.5]{Jaques2018SpectralChannels}, we know the maximal length of such a chain is $2D - 1$. 
    However, notice that $\mcM_{\phi} = \matrices$ implies $\phi:\matrices\to\matrices$ is unitary with respect to the inner product $\inner{\cdot}{\cdot}_\varrho$ by Lemma \ref{Lem:Mult_dom_of_prim_unital_Schwarz}.
    Thus, using the fact that $\phi$ is primitive and arguing as in Lemma \ref{Lem:Stab_mult_dom_of_primitive}, we conclude that $\matrices = \mcM_{\phi} = \mbC\mbI$, which is absurd.
    Thus, we see that $\kappa(\phi)\leq 2D - 2$, which concludes the proof. 
\end{proof}
From here, the proof of Theorem \ref{Thm:Main_theorem_compos} is routine.
\begin{proof}[Proof of Theorem \ref{Thm:Main_theorem_compos}]
    Let $\phi:\matrices\to\matrices$ be an arbitrary $2$-positive primitive map, and let $z$ be the Perron-Frobenius eigenvector of $\phi$.
    Define $\phi_z(a) = r^{-1}z^{-1/2}\phi(z^{1/2}az^{1/2})z^{-1/2}$.
    Then notice for $n\in\mbN$, $\phi^n$ is strictly positive if and only if $\phi_z^n$ is strictly positive. 
    Indeed, if $\phi^n$ is strictly positive, then for any $x\in\mbP_D\setminus\set{0}$, we have $z^{1/2}x z^{1/2}\in\mbP_D\setminus\set{0}$, hence $\phi^z\seq{z^{1/2}x z^{1/2}}>0$. 
    So, since $z^{-1/2}>0$ and $r>0$, we conclude $\phi_z^n(x) > 0$. 
    On the other hand, if $\phi_z^n$ is strictly positive, then for all $x\in\mbP_D\setminus\set{0}$, we have that 
    \begin{equation}
        \phi^n(x) = r^{-n} z^{1/2}\phi_z^n(z^{-1/2}xz^{-1/2})z^{1/2}
        >0,
    \end{equation}
    hence $\phi^n$ is strictly positive. 
    Therefore, $\primind{\phi} = \primind{\phi_z}$. 
    But $\phi_z(\mbI) = \mbI$ by construction, so since $\phi_z$ is $2$-positive (being the composition of 2-positive maps), by \cite{Choi1974AC-Algebras} we have that $\phi_z$ is a unital primitive Schwarz map.
    Thus, by Theorem \ref{Thm:Main_theorem_introduction}, we have that $\primind{\phi} = \primind{\phi_z}\leq 2(D-1)^2$.
\end{proof}
%
%
%


\section{Relation with \texorpdfstring{\cite[Conjecture 2]{Perez-Garcia2007MatrixRepresentations}}{l}}
We now discuss our result in relation to \cite[Conjecture 2]{Perez-Garcia2007MatrixRepresentations}, which we recall now.
For a collection $S = \set{v_i}_{i=1}^g\subset\matrices$ of matrices, recall the definition of the Wielength of $S$: 
\begin{equation}
    \texttt{Wie}\ell\seq{S}
        :=
    \inf\set{k\in\mbN
        \,\,:\,\,
    \operatorname{span} S^k = \matrices
    },
\end{equation}
where $S^k = \set{v_\sigma}_{\sigma\in\set{1, \dots, g}^k}$, where for $\sigma = \seq{i_1, \dots, i_k}\in \set{1, \dots, g}^k$, $v_\sigma := v_{i_k}\cdots v_{i_1}$.  
The number $\texttt{Wie}\ell(D)$ is physically interesting due to its relationship with the parent Hamiltonians in the theory of matrix product states \cite{Perez-Garcia2007MatrixRepresentations}.
It is conjectured that $\texttt{Wie}\ell(D)\leq O(D^2)$:
\begin{conj}[\texorpdfstring{\cite[Conjecture 2]{Perez-Garcia2007MatrixRepresentations}}{l}]\label{Conj:PG}
    For all $S\subseteq\matrices$ as above, $\texttt{Wie}\ell\seq{S}\leq O(D^2)$.
\end{conj}
If we let $\phi_S$ denote the map $\phi_S(a) = \sum_i v_ia v_i^*$, it is known that $\texttt{Wie}\ell(S)\geq \primind{\phi_S}$, where this inequality may be strict \cite{Sanz2010AInequality}. 
That is, a bound on $\texttt{Wie}\ell\seq{S}$ gives a bound on the index of primitivity for completely positive maps. 
Notice, however, that a bound on $\texttt{Wie}\ell(S)$ does \textit{not} imply a bound on the index of primitivity for an arbitrary unital primitive Schwarz map, since not all Schwarz maps are completely positive. 
Thus, a resolution of Conjecture \ref{Conj:PG} does not imply our Theorem \ref{Thm:Main_theorem_compos}.
The preprint \cite{Shit} proposes a resolution of Conjecture \ref{Conj:PG}, which gives an improvement of our bound on $q$ in the particular case of completely positive maps. 
We suggest a different method for proving Conjecture \ref{Conj:PG} that makes use of our Theorem \ref{Thm:Main_theorem_compos} by posing the following question: 
How far from $\primind{\phi_S}$ can $\texttt{Wie}\ell(S)$ possibly be?
A bound on the difference between these two quantities may yield a proof of Conjecture \ref{Conj:PG}.
This seems to be an interesting avenue for further research.

\appendix 

\printbibliography

\end{document}